\documentclass[12pt,hidelinks]{article}
\usepackage[english]{babel}
\usepackage{amsmath,amsthm,amssymb}
\usepackage{mathabx}
\usepackage[top=1.5in, bottom=1.5in, left=1.5in, right=1.5in]{geometry}
\usepackage{graphicx}
\usepackage{amsfonts}
\usepackage[utf8]{inputenc}
\usepackage{tikz}
\usepackage{tikz-cd}
\usetikzlibrary{trees}
\usetikzlibrary{matrix}
\usepackage{verbatim}
\usepackage{appendix}
\usepackage{natbib}
\usepackage{datetime}
\usepackage{setspace}
\usepackage{bookmark}
\usepackage{quoting} 

\usepackage{array}
\usepackage{booktabs}
\usepackage[hyperpageref]{backref}

\usetikzlibrary{decorations.markings}

\renewcommand*{\backref}[1]{}
\renewcommand*{\backrefalt}[4]{[%
    \ifcase #1 Not cited.%
          \or p.~#2%
          \else p.~#2%
    \fi%
    ]}

\newtheorem{theorem}{Theorem}

\theoremstyle{definition}

\theoremstyle{property}

\theoremstyle{proof}

\theoremstyle{definition}
\newtheorem{definition}{Definition}
\theoremstyle{remark}

\numberwithin{equation}{section}

\begin{document}
	\onehalfspacing
	\title{No-harm principle, rationality, and Pareto optimality in games}
	\author{Shaun Hargreaves Heap\footnote{Department of Political Economy, King's College London, London, UK. E-mail: s.hargreavesheap@kcl.ac.uk} \and Mehmet S. Ismail\footnote{Department of Political Economy, King's College London, London, UK. E-mail: mehmet.ismail@kcl.ac.uk}}

	\date{\today}
	\maketitle
	\begin{abstract}

Mill's classic argument for liberty requires that people's exercise of freedom should be governed by a no-harm principle (NHP). In this paper, we develop the concept of a no-harm equilibrium in  $n$-person games where players maximize utility subject to the constraint of the NHP. Our main result is in the spirit of the fundamental theorems of welfare economics. We show that for every initial `reference point' in a game the associated no-harm equilibrium is Pareto efficient and, conversely, every Pareto efficient point can be supported as a no-harm equilibrium for some initial reference point.

	\end{abstract}
	
	\newpage

 \section{Introduction}
\label{sec:intro}

	\begin{quoting}
		``It's a pretty long-standing principle that goes back to the 19th century that you are free to do things but not if you inflict harm on others.''
		
		(Nick Clegg explaining Facebook's decision to ban Donald Trump) 
	\end{quoting}
	
	Nick Clegg reminds us that the sense of liberty which animates the liberal democracies of North America, Europe and elsewhere is not a free-for-all.  At least since John Stuart Mill's \textit{On Liberty}, the writ of individual freedom in a society founded on the principle of liberty is not without constraint: people's freedom does not extend to harming others. For Mill, the failure to satisfy this condition supplied the only reason for interfering in what people wish to do: 
	\begin{quoting}
		``...the only purpose for which power can be rightfully exercised over any member of a civilised community, against his will, is to prevent harm to others.'' \citep{mill1859}
	\end{quoting}
	This may be the only reason for constraining individual action according to Mill, but, as Nick Clegg reminds us, it is, nevertheless, a constraint. In this paper, we examine how the introduction of such a no-harm principle (NHP) as a constraint on action in normal form games affects the analysis of rational play in these interactions. In particular, we develop the concept of a no-harm equilibrium where players maximize utility subject to the constraint of the NHP. Our main result is in the spirit of the fundamental theorems of welfare economics, where consumers maximize utility subject to their budget constraints. We show that for every initial `reference point' in a game the associated no-harm equilibrium is Pareto efficient and, conversely, every Pareto efficient point can be supported as a no-harm equilibrium for some initial reference point. 
	
	This is a striking theoretical result because it is in marked contrast to the well-known insight, for example in the conventional analysis of the prisoners' dilemma, that the Nash equilibrium of normal form games need not be Pareto efficient. Our result is not only potentially important for this reason, it also has policy implications. In the conventional wisdom, situations like a prisoners' dilemma provide a \textit{prima facie} case for policy  intervention to secure Pareto improvements. Our contrary result suggests something different. It is in the spirit of Coase's theorem in the sense that it is not the fact that there are, say, externalities in some settings that supplies the \textit{prima facie}  grounds for policy intervention. For \cite{coase1960}, it is the presence of non-negligible transaction costs. In our case, it  would be the failure, in practice, of agents to be  guided by the NHP that establishes the  \textit{prima facie} grounds for intervention. The NHP might not be operational in practice for a variety of reasons. For example, people may subscribe to a free-for-all sense of liberty and not the NHP version; or the NHP may be too cognitively demanding to guide citizens in a liberal society. Alternatively, the NHP may be rendered inoperable because there are disagreements over the relevant `reference point' or how to define `harm'. It is these conditions, our result suggests, that should trigger a \textit{prima facie} policy interest.
	
	Of course, our contrary result may be theoretically arresting and point to a different policy agenda, but it is only really important in so far as the NHP is thought to be relevant to the analysis of action in normal form games. The premise of the NHP as a constraint on action, in other words, has to be plausible. Nick Clegg's appeal to the principle is an illustration of its possible relevance in this respect. More compelling, perhaps, is the fact that the business of the judiciary in contemporary liberal societies frequently involves deciding when one person's exercise of liberty may or may not be reasonably said to have conferred a `harm' upon another. What counts as a `harm' is, of course, naturally controversial and this is why the courts get involved. The point, however, about the courts' involvement in such matters is that it arises because they are upholding the NHP. In other words, the NHP, and not a free-for-all, is constitutive, notably through the activities of the courts, of what freedom means in liberal democratic societies. As a result, the NHP might plausibly be thought to apply to all who voluntarily live in such societies. 
	
	This constitutive role of NHP in the liberal democratic understanding of freedom might be granted, but it does not necessarily mean that game theory should adopt the principle when analyzing what players should rationally do in games. It might be argued, for instance, that the judiciary identifies actions that satisfy this principle and so, when game theory's available actions are legal, all actions satisfy the principle and there is no need to apply a further NHP test in the analysis of what players will freely and rationally do. Game theory, though, does not typically require that actions available to players be legal. If game theory only considered legal actions, it would be unable to analyze when people follow the law (or more generally  uphold a contract) or indeed when they engage in civil disobedience. From this perspective it might, nevertheless, be alternatively argued that courts put a price on such harms and so the force of the NHP is in effect encoded in the players' pay-offs. However, what is or is not legal is only infrequently tested by courts and so the precise legal price of a harm is frequently not available. People in liberal societies, therefore, may need to anticipate and apply the NHP themselves: either because they need to anticipate what the courts might say if called upon or because they simply value freedom and so are guided by its entailment in liberal society. Facebook, for example, was not executing the court's judgement when banning Donald Trump. Nick Clegg, on behalf of Facebook, was either saying something like `Facebook values freedom and this means being guided by the NHP'; or he was anticipating what a court might decide when applying this principle, as it does when determining what freedom requires in a liberal society. 
	
	For these reasons, the NHP appears relevant to the analysis of action in games. Our paper, therefore, proceeds, as follows. In the next section, we provide an informal introduction to the challenges posed by introducing the NHP and sketch how we respond to them. We then set out our approach formally in section~\ref{sec:NHP}; and section~\ref{sec:existence} gives the results associated with the inclusion of NHP in three theorems. We reflect on these results in section~\ref{sec:illustrations} with some further illustrations and then discuss their relation to the literature in section~\ref{sec:discussion}. Section~\ref{sec:conclusion} concludes the paper.
	
	\section{An informal sketch of the challenges of the no-harm principle and our approach}
	
	The NHP does not permit a person to take an action that causes `harm' to others. Several challenges arise when deciding how to represent this principle as a constraint on actions in games. In this section we offer an informal sketch of how we respond to these challenges by making four key assumptions.

	Game theory has one advantage over the world that the judiciary addresses when trying to decide whether someone's actions cause a `harm' to another: game theory deals with interactions where the pay-offs, usually captured by utility numbers, to each player in each outcome are given.  With players' interests captured in this way by their pay-offs, it seems natural and uncontentious to say that a person suffers a `harm' when their pay-offs are reduced. We deliberately use utility numbers to capture ``pay-offs'' in what follows because this allows for an encompassing definition of possible harms. They could be psychological or symbolic  as much as material whereas to use \$ material pay-offs would be to restrict the concept of a harm to a \$ loss alone. 
	
	If a harm arises when a person’s action reduces another’s utility pay-off, a question naturally arises: a reduction relative to what? What is the reference point pay-off for judging whether the action causes a harm?  This is the first modelling challenge. Since games contain all the relevant available actions for players in the setting captured by the game, we assume that the reference pay-offs must be given by one of the outcomes in the game. We make no argument over which outcome should be used. Any outcome might serve as the reference. Instead, we seek to characterize in general terms the equilibria that result when players take any of the outcomes in the game as the shared reference point. The specific attributes of an equilibrium that satisfies the NHP may depend on the actual reference point, but we are interested with any general properties of such equilibria. 
	
	The next challenge arises because the NHP requires that an individual's action should not cause a harm to any other person and outcomes in normal form games typically result from the joint actions of several players. We need, therefore, some way of making sense of how an individual produces an outcome in a game and so judge whether that individual's action causes a harm.  Our approach follows \cite{brams1994}. It takes one outcome in the game as the reference point (one might think of this as the status quo) and builds an extensive form game on the basis of the possible sequential player deviations from this reference point. That is, the first player in this sequence must decide between `passing' their turn, `staying' at the reference outcome, and `moving' to an alternative outcome by changing their action so as to produce the alternative; and each player thereafter decides in this extensive form game between passing their turn, staying at the outcome they have inherited from the previous decisions of others in this branch of the extensive form game  and moving to an alternative one by changing their action. To fix ideas here consider the Prisoners dilemma game below.  
 
\begin{table}[h!]
	\begin{center}
 \setlength{\aboverulesep}{0pt}
\setlength{\belowrulesep}{0pt}
		\begin{tabular}{|c|c|c|}
			\midrule
			& C & D \\
			\midrule
			C & $3,3$ & $1,4$ \\
			\midrule
			D &	$4,1$ & $2,2$ \\
			\midrule
		\end{tabular}
	\end{center}\end{table}
	Suppose DD is the reference outcome and Row is the first player to consider a deviation. The extensive form game is constructed as follows. Row can either `pass' their turn to Column, `stay' at DD, or `move' to CD by changing their action to C. If Row chooses to `pass' at DD, then, Column faces the same decision at the second node of this extensive form game as the one we have just considered for Row. If Row chooses to `stay' at DD, then DD remains the outcome. If  Row chooses to `move' to CD, Column at the second decision node in the extensive form game must decide between passing, staying at CD  and moving to CC by changing their action from D to C; if Column chooses to pass at CD then Row decides between passing, staying at CD, and moving to DD; and so on. Figure~\ref{fig:PD_illustration1} captures these early decision nodes in the extensive form game based on DD as the reference outcome.
	
	\begin{figure}
		\centering
		\resizebox{.93\textwidth}{!}{
		\begin{tikzpicture}[font=\footnotesize,edge from parent/.style={draw,thick}]
		\tikzstyle{solid node}=[circle,draw,inner sep=1.2,fill=black];
		\tikzstyle{hollow node}=[circle,draw,inner sep=1.2];
		\tikzstyle{level 1}=[level distance=15mm,sibling distance=50mm]
		\tikzstyle{level 2}=[level distance=15mm,sibling distance=25mm]
		\tikzstyle{level 3}=[level distance=15mm,sibling distance=15mm]
		\node(0)[hollow node]{}
		child{node[solid node]{}
			child{node[solid node]{}
				child{node[below]{} edge from parent[shorten >=0.7cm, dotted, thick] node[left]{}}
				edge from parent[black, solid, thick] node[above left]{$C$ moves}
			}
			child{node[solid node][below]{}
				child{node[below]{} edge from parent[shorten >=0.7cm, dotted, thick] node[left]{}}
				edge from parent[black, solid, thick] node[below left]{$C$ stays}
			}
			child{node[solid node][below]{}
				child{node[below]{} edge from parent[shorten >=0.7cm, dotted, thick] node[left]{}}
				edge from parent[black, solid, thick] node[above right]{$C$ passes}
			}
			edge from parent[black, solid, thick] node[above left, yshift=-3]{$R$ moves}
		}
		child{node[solid node, below]{}
			child{{}
			edge from parent[shorten >=0.9cm, dotted, thick] node[above left]{}
			}		
			edge from parent[thick] node[right]{$R$ stays}
		}
		child{node[solid node]{}
			child{{}
				edge from parent[shorten >=0.9cm, dotted, thick] node[above left]{}
			}
			edge from parent[thick] node[above right, yshift=-4]{$R$ passes}
		}
	;
		\node[above,yshift=2]at(0){DD};
		
		\node[above,yshift=2]at(0-1){CD};
		\node[right,yshift=3]at(0-2){DD};
		\node[right,yshift=3]at(0-3){DD};
		
		\node[left,yshift=3]at(0-1-1){CC};
		\node[right,yshift=3]at(0-1-2){CD};
		\node[right,yshift=3]at(0-1-3){CD};
		\end{tikzpicture}
		}
		\caption{Illustration of early decision nodes in the extensive form game based on DD as the reference outcome}
		\label{fig:PD_illustration1}
	\end{figure}
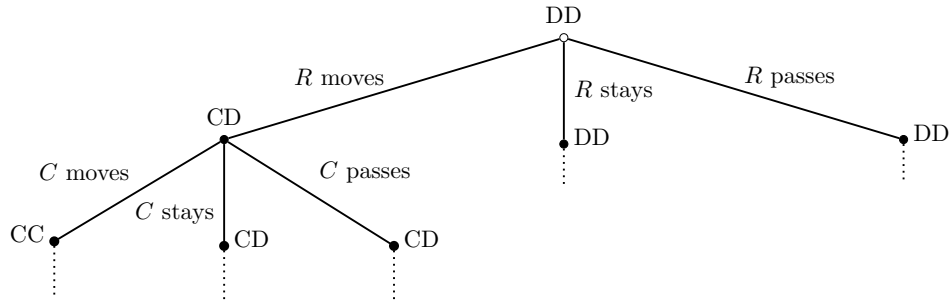
	
	The virtue for our purpose of adopting Brams's Theory of Moves in this way is that each outcome in this extensive form game now occurs through an individual action at some decision node and this allows us to identify how an individual's action can be said to cause a particular outcome.  Thus, we can say in our illustration that, if Row begins by `moving' to CD through a `move' to C from the reference point of DD, then Row has caused CD at this point in the extensive form game. Likewise, if Row initially decides to `stay' at DD, they have caused DD. However, when Row `passes' at DD, they do not cause DD because they have exempted themselves from decision making by passing the choice to the Column player. (In effect, giving players an option to `pass' endogenises the player order for these sequential deviations and so has no influence in a $2\times 2$ symmetric game like the Prisoners' Dilemma.  In more complicated $n$-player games, `pass' has further technical function of enabling every outcome to be visited through a sequence of player deviations.) 
	
	Brams (1994) assumes farsighted rationality and solves by backward induction the extensive form game created by his procedure of sequential deviations. He calls the outcome of his proposed sequential procedure a non-myopic equilibrium (see also Brams and Wittman, \citeyear{brams1981}, and Kilgour, \citeyear{kilgour1984}).  We adopt the same approach of applying the subgame perfect equilibrium solution concept but we introduce the NHP as a constraint on play in the extensive form game and so call the outcome a no-harm equilibrium. Thus, to return to the prisoners' dilemma with DD as the initial reference outcome, we ask players whether the subgame perfect equilibrium of the NHP constrained extensive form game based on DD is DD or some other outcome. If it is DD, then DD is the no-harm equilibrium outcome of the game. If it is not, then the no-harm equilibrium with DD as the reference outcome is whatever the subgame perfect equilibrium is in this NHP constrained extensive form game.  The next challenge is, therefore, to represent the NHP in this extensive form game. This requires three further assumptions. 
	
	One is innocuous in the sense that a well-defined finite extensive form game requires a set of terminal nodes. One part of how we do this is by saying that when both players decide to `stay' at an outcome, then this outcome is implemented. The intuition behind this assumption is that the first decision to `stay' is like a `proposal' to implement this outcome and the second decision to `stay' amounts to an acceptance of this `proposal'. This naturally produces a terminal node for some branches of the extensive from game. Likewise, if both players decide to pass, then this naturally produces a terminal node. However, we also need to prevent the `moving' branches of the tree creating what are infinite cycles through the possible outcomes in the game. We do this by preventing a player repeating the same action (`move' or `stay') a certain number of times (say, $k$ times where our results hold irrespective of $k$). Suppose, for example, we follow the branch in the extensive form game that begins with Row `moving' to CD, Column next `moves' to CC, Row then `moves' to DC and Column `moves' to DD. When $k=1$, Row cannot `move' to CD again. Row can either `stay' or `pass' but if they choose to `stay', they cannot choose to `stay' again at DD. Thus, DD is the terminal outcome either because both players `stay' or both `pass' at DD  for  this truncated move branch of the extensive form game.
	
	The question arises as to what outcome should be implemented if players reach a terminal node through joint `passing'---either because the truncation $k$ rule has been triggered or because players have both chosen to `pass' at an earlier point in the extensive form game. The terminal outcome cannot be said to have been chosen by the players in these circumstances because neither has decided to stay at this terminal outcome. We assume, therefore, that its reference point is implemented because no other outcome along the path to this terminal outcome, including the terminal outcome itself, has been consciously endorsed by both players through `stay' decisions.
	
	The next assumption embodies the NHP. We say that in so far as someone's choice of action contributes to causing an outcome \underline{that is implemented}, then they are only permitted to take that action if the resulting outcome does not harm other players. An individual can only contribute to causing an outcome that is implemented by deciding to `stay' at that outcome. Of course, it takes more than one decision to `stay' for an outcome actually to be implemented. But in so far as one individual could contribute to an outcome being implemented, they would do so by deciding to `stay' at that outcome.  Thus, we apply the NHP to an individual's decision to `stay'. It does not apply to a `move' decision for the same reason: the condition for implementation is two `stay' decisions and so only `stay' can be said to contribute to causing and potentially implementing an outcome. To illustrate, in the extensive form game that begins with DD as the reference outcome in the prisoners' dilemma, Row `staying' at DD satisfies the NHP and so does Row `moving' to CD (because  NHP only applies to `stay' decisions). However, Column's option to `stay' at CD would not satisfy the NHP (because it harms Row), but moving to CC does (because NHP does not apply to `move'). Thus, the beginning of the NHP constrained extensive form game looks like Figure~\ref{fig:PD_illustration2} (with Column stays at CD  faded out as an option as compared with the unconstrained version in Figure~\ref{fig:PD_illustration1}).

	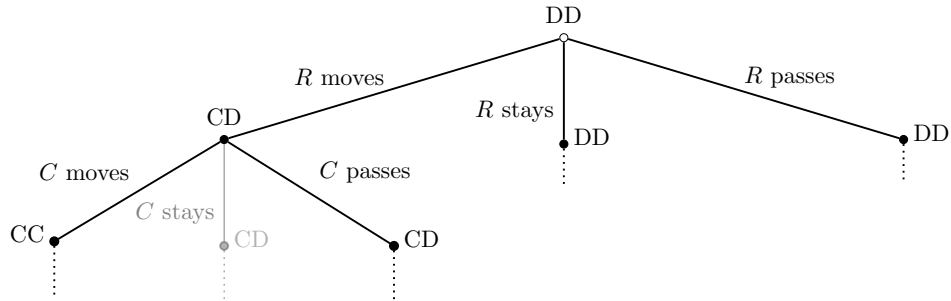
\begin{figure}
		\centering
		\resizebox{.93\textwidth}{!}{
		\begin{tikzpicture}[font=\footnotesize,edge from parent/.style={draw,thick}]
		\tikzstyle{solid node}=[circle,draw,inner sep=1.2,fill=black];
		\tikzstyle{hollow node}=[circle,draw,inner sep=1.2];
		\tikzstyle{level 1}=[level distance=15mm,sibling distance=50mm]
		\tikzstyle{level 2}=[level distance=15mm,sibling distance=25mm]
		\tikzstyle{level 3}=[level distance=15mm,sibling distance=15mm]
		\node(0)[hollow node]{}
		child{node[solid node]{}
			child{node[solid node]{}
				child{node[below]{} edge from parent[shorten >=0.7cm, dotted, thick] node[left]{}}
				edge from parent[black, solid, thick] node[above left]{$C$ moves}
			}
			child{node[opacity=0.3,solid node][below]{}
				child{node[below]{} edge from parent[shorten >=0.7cm, dotted, thin] node[left]{}}
				edge from parent[opacity=0.5, thin] node[below left]{$C$ stays}
			}
			child{node[solid node][below]{}
				child{node[below]{} edge from parent[shorten >=0.7cm, dotted, thick] node[left]{}}
				edge from parent[black, solid, thick] node[above right]{$C$ passes}
			}
			edge from parent[black, solid, thick] node[above left, yshift=-3]{$R$ moves}
		}
		child{node[solid node, below]{}
			child{{}
			edge from parent[shorten >=0.9cm, dotted, thick] node[above left]{}
			}
			edge from parent[thick] node[below left]{$R$ stays}
		}
			child{node[solid node]{}
			child{{}
				edge from parent[shorten >=0.9cm, dotted, thick] node[above left]{}
			}
			edge from parent[thick] node[above right, yshift=-4]{$R$ passes}
			}	
	;
		\node[above,yshift=2]at(0){DD};
		
		\node[above,yshift=2]at(0-1){CD};
		\node[right,yshift=3]at(0-2){DD};
		\node[right,yshift=3]at(0-3){DD};
		
		\node[left,yshift=3]at(0-1-1){CC};
		\node[opacity=0.3, right,yshift=3]at(0-1-2){CD};
		\node[right,yshift=3]at(0-1-3){CD};		
		\end{tikzpicture}
		}
		\caption{The beginning of the NHP constrained extensive form game where DD is the reference outcome.}
		\label{fig:PD_illustration2}
	\end{figure}
	
	This illustration gives an immediate insight into how individual farsighted rationality and the NHP might combine to make CC the no-harm equilibrium in this extensive form game. The pursuit of the best option for Column at the second CD node by `staying' is precluded by the NHP. At any later decision node on this branch of the game, Row will not be able to `stay' at DC for the same reason and so the only possible terminal options along this branch will be DD (either through Row's initial `stay' decision or through sequential `moves' to this terminal node followed by mutual `passes' leading to its reference point, DD, being implemented) and CC (through mutual `stay' decisions at this point). Thus either DD or CC will be implemented and farsighted rationality will secure CC. The details are, of course, a bit more complicated. One comment is worth making, nevertheless. 
	
	Our NHP principle deliberately does not embody farsighted rationality. We introduce farsighted rationality as a separate assumption. The application of farsighted rationality requires an assumption of common knowledge of rationality and there is no reason to bind the NHP to such an assumption. A `harm' is a `harm' whether the other player is rational or not. The most that can be said in such circumstances with respect to whether a player's action causes a harm, is: does that action by itself do as much as any individual can do to causing an outcome to be implemented that harms someone else? That is, do they choose to `stay'.  
	
	Two final details in our approach are worth noting at this stage. First, we require only two stay decisions at a particular outcome for it to be implemented in an `$n$-person' game. This is because we wish to avoid building in Pareto improvements through some version of unanimity rule that requires everyone to agree on some outcome before it is implemented. Our condition for implementation is, therefore, in general, much weaker than unanimity. However, in a two-person game it does amount to unanimity and this supplies another part of the intuition behind why CC emerges as the no-harm equilibrium in the prisoners' dilemma illustration. In section~\ref{sub:termination}, we extend our model to the case in which any player can unilaterally implement the outcome on their turn.  These and other natural possible modifications do not affect our results.
	
	Second, a final key assumption for more complicated games than the two person prisoners' dilemma is that if say player $i$ decides to stay at an outcome but the next player $j$ decides to reject this proposal by moving to another outcome, then the reference point is updated to the outcome proposed by $i$ through their decision to `stay'. In effect, this outcome has been endorsed by $i$, it satisfies the NHP for $i$ and it could have been implemented by $j$; so, it is natural to use this as the (new) reference point for judging future deviations. Thus, in general, the decision to `stay' is also a decision to change the reference point and this can only be done by a player when to do so satisfies the NHP.

	\section{The no-harm principle in non-cooperative games}
	\label{sec:NHP}
	
	\subsection{The setup}

	Let $G=(A_i,u_i)_{i\in  N}$ denote a normal form game in which $N=\{1,2,…,n\}$ denotes a society whose members are called players, $A_i$ finite pure action set of player $i$, $u_i:A\rightarrow \mathbb{R}$ player $i$'s Bernoulli utility function representing a strict ranking over the set of action profiles $A=\times_{i\in  N} A_i$. Let $a=(a_1, a_2, ..., a_n)\in  A$ denote a pure action profile in game $G$.\footnote{Our definitions can be extended to the games with mixed strategies in a straightforward way. We keep the current framework for its simplicity.} As is standard in normal form games, every action profile is associated with an outcome (i.e., a pay-off profile) and vice versa. We use the terms ``action profile'' and ``outcome'' interchangeably.
	
	A profile $a$ Pareto dominates $a'$ if for all $i$, $u_i (a)\geq u_i (a')$ with at least one strict inequality. A profile is called Pareto optimal or efficient if there is no other profile that Pareto dominates it. A profile is called weakly Pareto optimal if there is no other profile in which everyone is strictly better off.

	Fix a game $G=(A,u)$ and action profile $a_0\in A$. For a natural number $k\in \mathbb{N}^+=\{1,2,...\}$, and a player function $I$, we define an associated extensive form game with perfect information denoted by $\Gamma(a_0, k, I)=(N, X , I, u, \Sigma, H)$. We refer to this extensive form game as simply ``$\Gamma$''. The interpretation of $\Gamma$ is that starting from $a_0$, each player $i$ sequentially decides to `stay', `move' or `pass' in game $G$ until the play terminates.\footnote{For a standard textbook on extensive form games, see, e.g., \cite{osborne1994}.}

	Let $X$ denote a game tree, $x\in X$ a node in the tree, $x_0$ the root of the game tree, and $z\in Z$ a terminal node, which is a node that is not a predecessor of any other node.
	
	\subsubsection{Player function}
	
	Let $[x_m]=\{x_0, x_1, x_2, ..., x_m\}$ denote the \textit{path of play} between node $x_0$ and node $x_m \in X$ where for every $j=0,1,...,m-1$, $x_{j+1}$ is an immediate successor of $x_{j}$. Let $I:X\to N$ be the player function, where $I(x)$ gives the ``active'' player who moves at node $x$. The only restriction we impose on the player function is the following. For player $i$, let $[x_m]_i=|\{x'\in [x_m]| I(x')=i\}|$, i.e., the number of times player $i$ is active during the path of play $[x_m]$. We assume that for every player $i$, every player $j\neq i$, and every path of play $[x_m]$, $|[x_m]_i-[x_m]_j|\leq 1$. In other words, if a player has been active $\bar{m}$ times in some path of play, then every other player should have been active at least $\bar{m}-1$ times. This assumption ensures that every player has more or less the same number of moves to play on every path. Beyond this assumption, notice that the order in which players take turn is not fixed and at every non-terminal node $x$ the next active player may depend on the particular action chosen by player $I(x)$.
	

	\subsubsection{Actions, strategies, and information sets}

	Let $h\in H$ denote an information set, which is a singleton. With a slight abuse of notation, an information set $h$ at node $x$ is denoted by $x$, i.e., $h=x$. A \textit{subgame} $\Gamma|x$ of a game $\Gamma$ is the game $\Gamma$ restricted to an information set $h=x$ and all of its successors in $\Gamma$. 
	
	Next, we introduce $S:X \rightarrow A$, called the \textit{state function}, that maps each node $x\in X$ to an action profile $a\in A$. We define $S$ by induction. The state at the root $x_0$ of the game is defined as $a_0$, i.e. $S(x_0)=a_0$, which is the action profile in $G$ where the extensive form game $\Gamma$ starts. Let $x\in X$, $x\neq x_0$, be a node, $x'\in X$ the immediate predecessor of $x$, $i=I(x')$, and $a'_i$ player $i$'s action that leads to node $x$. Assume that $S(x')=a$. Then, define $S(x)=(a'_i,a_{-i})$. In other words, at every node $x$, the state $S(x)$ is given by the action profile $(a'_i,a_{-i})\in A$ such that player $i=I(x')$ changes only the $i$'th component of the state at $x'$.

	Let $A_i(x)$ denote the set of pure actions of player $i$ at $x$. For each $x\in X$, $A_i(x)$ is defined as follows. First, define $X'(a_i,x)=\{x'\in X| a_i~\text{is chosen at}~x',$ $S(x')=S(x),\text{and}~x'~\text{is a predecessor of}~x \}$. Then, $A_i(x)=\{a_i\in A_i|~|X'(a_i,x)|$ $<k \}\cup \{p\}$, where $p$ stands for pass. For example, suppose that $k=1$, $i=I(x_0)=I(x'')$, $x_0\neq x''$, and $S(x_0)=S(x'')$. If $i$ chooses $a_i$ at $x_0$, then $a_i\notin A_i(x'')$ because $k=1$ implies that $a_i$ can be chosen at state $S(x_0)$ only once. Unless otherwise stated, we assume that $k=1$ in the examples throughout the paper. This means once a player returns to a state, they cannot make the same decision as they made last time they were at this state. This is the assumption that prevents infinite cycling through always choosing to `move' at each decision node and/or through one player staying and the other(s) passing, repeatedly at the same state. When this truncation rule is binding at a decision node $x$, players can only choose to pass at $x$.

	
	Let $A'_i=\bigcup_{x\in X_i}A_i(x)$ denote player $i$'s set of all pure actions where $X_i$ is player $i$'s set of all information sets. Let $\Sigma_i=\bigtimes_{x\in X_i} A_i(x)$ denote the set of all pure strategies of $i$ where a pure strategy $\sigma_i\in \Sigma_i$ is a function $\sigma_i:X_i\rightarrow A'_i$ satisfying $\sigma_i(x)\in A_i(x)$ for all $x\in X_i$. Let $\sigma\in \Sigma$ denote a pure strategy profile and $u_i(\sigma)$ its (Bernoulli) utility for player $i$.

	Let $a'_i\rightarrow x'$ denote player $i$'s action $a'_i\in A_i(x)$ that leads to node $x'\in X$.
	Let $[\sigma]=\{x\in X|\sigma_i(x')\rightarrow x ~\text{for some}~i\in N, x'\in X\}\cup \{x_0\}$ be the path of play of $\sigma$ and $\overline{[\sigma]}$ be the terminal node in $[\sigma]$.

	
	\subsubsection{Reference points, terminal nodes, and utility functions}
	\label{sub:reference_points}
	
	For a given strategy profile $\sigma$, let 
	$$R(\sigma)=\{a\in A |x\in [\sigma], \sigma_i(x)=a_i\in A_i(x), a=S(x)\}\cup \{a_0\}$$ be the set of all \textit{reference points} of $\sigma$. In other words, a state is called a reference point if the player who acts at the associated node ``stays'' at it: that is chooses not to change it. The initial reference point $a_0$ is included in $R(\sigma)$.
	
	Note that given a profile $\sigma$, for every decision node $y \neq x_0$ there is a unique reference point $a_y\in R(\sigma)$ where $S(x)=a_y$ for some predecessor $x$ of $y$. The unique reference point at $x_0$ is $a_0$ by definition. Thus, we can refer to \textit{the} reference point at every node $y\in X$.
	
	We next define `off-path' reference points. Let $y\neq x_0$ be a non-terminal node. The set of  reference points of $(\sigma|y)$, denoted by $R_{|y}(\sigma)$, is defined as $R(\sigma|y)$ except that $a_0$ is replaced with $a_y$, which is the reference point at $y$, $a_y\in R(\sigma)$. The intuition is that if we restrict a strategy profile $\sigma$ to a node $y$, then the initial reference point of $(\sigma|y)$ should be $a_y$ and not necessarily $a_0$.
	
	Game $\Gamma$ comes to an end under two situations. First, let $x'$ be a node and $x$ be a (not necessarily immediate) successor of $x'$ such that $I(x')=i$, $I(x)=j$, $j\neq i$, and the reference point at $x$ is $x'$ where $S(x')=S(x)=b$. If player $i$ stays at $b$ by choosing $b_i$, making $b$ the reference point, and player $j\neq i$ also stays at $b$ by choosing $b_j$, then node $x$ is called a \textit{terminal node}. Second, let $\{x^1, x^2, ..., x^n, x^{n+1}\}$ be a path of play such that for every player $i\in N$ there exists $x^m$ such that $i=I(x^m)$ where $n\geq m \geq 1$, and for every $m$, $x^{m+1}$ is an immediate successor of $x^{m}$. Node $x^{n+1}\in X$ is called a \textit{terminal node} if every $i$ chooses $p$ (i.e., pass) at $x^{m}$. In plain words, the game terminates if either (i) two distinct players choose to stay at a state (the first one is like a `proposal' to implement this state and the second one amounts to an acceptance of this `proposal'), or (ii) every player consecutively passes their turn.  
	
	Let $\sigma \in \Sigma$ be a strategy profile, $\overline{[\sigma]}=z$ its terminal node, and $a\in R(\sigma)$ the reference point at $z$. We define the \textit{outcome} of $\sigma$ as $a$.
	With slight abuse of notation we use the same utility function for $u_i(\sigma)$ and $u_i(a)$ because their outcomes, and hence their utilities are the same. In summary, for every player $i$, $u_i(\sigma)=u_i(a)$, where $a$ is the reference point at $z=\overline{[\sigma]}$. Put simply, the reference point at the terminal node is implemented as the outcome of the relevant strategy profile under both (i) and (ii) above. The reason the reference point at the terminal node is implemented as the outcome in condition (ii) is that no other state from the reference point to the terminal node has been endorsed by any player through a `stay' decision and that players have consciously chosen not to stay at the terminal node. As mentioned earlier, (ii) is in part a technical condition that prevents infinite cycling. In section~\ref{sec:existence} (proof of Theorem~\ref{thm:pareto}), we show that there is always an NHE whose outcome is attained under condition (i).
	
	\subsubsection*{An illustrative example}
	
	To illustrate our notation, we return to the prisoners' dilemma (PD). Let $\Gamma(a_0, k, I)$ be the extensive form game that begins this time with the reference point $a_0=$ CC. Assume that Row (player 1) moves first, Column (player 2) moves second, and this sequential order strictly alternates irrespective of players' choices. Starting from CC players might end up at DD if they play as follows (see Figure~\ref{fig:illustrative1}). Row unilaterally switches their action to D, hence `moving' to DC. Column then moves to DD, where Row chooses D to `stay' which makes DD the new reference point. Column also stays at DD, where both players receive pay-offs of (2, 2). If, instead of staying, both Row and Column choose to pass at DD, then the implemented outcome would be the reference point at this node, which is CC. The difference between passing and staying is that passing changes the order of play but does not change the reference point.
	
	Note that $k=1$ implies that a `cycle' cannot be repeated. Suppose, for example, that Row moves to DC from CC, Column moves  to DD, Row moves to CD, and Column moves back to CC. Then, if $k=1$ Row cannot choose D again at CC. Row can now only either stay or pass at CC.
	
	We have not so far imposed any restrictions on the players' choices such as `rationality' or `no-harm principle'. We next introduce the NHP.

			\begin{figure}
				\centering
				\resizebox{.93\textwidth}{!}{
				\begin{tikzpicture}[font=\footnotesize,edge from parent/.style={draw,thick}, decoration={
							markings,
							mark=at position 0.5 with {\arrow{>}}}
				]
				\tikzstyle{solid node}=[circle,draw,inner sep=1.2,fill=black];
				\tikzstyle{hollow node}=[circle,draw,inner sep=1.2];
				\tikzstyle{level 1}=[level distance=15mm,sibling distance=50mm]
				\tikzstyle{level 2}=[level distance=15mm,sibling distance=25mm]
				\tikzstyle{level 3}=[level distance=15mm,sibling distance=15mm]
				\node(0)[hollow node]{}
				child{node[solid node]{}
					child{node[solid node]{}
						child{node[below]{} edge from parent[shorten >=0.7cm, dotted, thick] node[left]{}}
						edge from parent[black, solid, thick] node[above left]{C}
					}
					child{node[solid node][below]{}
							child{node[solid node]{}
								child{node[below]{} edge from parent[shorten >=0.7cm, dotted, thick] node[left]{}}
								edge from parent[black, solid, thick] node[above left]{C}
							}
						child{node[solid node][below]{}
								child{node[solid node]{}
									child{node[below]{} edge from parent[shorten >=0.7cm, dotted, thick] node[left]{}}
									edge from parent[black, solid, thick] node[above left]{C}
								}
								child{node[below]{$\begin{pmatrix} 2\\2\end{pmatrix}$}
									edge from parent[black, solid] node[left]{D}
								}
								child{node[solid node][below]{}
									child{node[below]{} edge from parent[shorten >=0.7cm, dotted, thick] node[left]{}}
									edge from parent[black, solid, thick] node[above right]{$p$}
								}					
							edge from parent[black, solid] node[left]{D}
						}
							child{node[solid node][below]{}
								child{node[below]{} edge from parent[shorten >=0.7cm, dotted, thick] node[left]{}}
								edge from parent[black, solid, thick] node[above right]{$p$}
							}					
						edge from parent[black, solid] node[left]{D}
					}
					child{node[solid node][below]{}
						child{node[below]{} edge from parent[shorten >=0.7cm, dotted, thick] node[left]{}}
						edge from parent[black, solid, thick] node[above right]{$p$}
					}
					edge from parent[black, solid] node[above left, yshift=-3]{D}
				}
				child{node[solid node, below]{}
					child{{}
			        edge from parent[shorten >=0.9cm, dotted, thick] node[above left]{}
			        }
					edge from parent[thick] node[right]{C}
				}
				child{node[solid node]{}
					child{{}
						edge from parent[shorten >=0.9cm, dotted, thick] node[above left]{}
					}
					edge from parent[thick] node[above right, yshift=-4]{$p$}
				}
			;
				\node[above,yshift=2]at(0){1};
				
				\node[above,yshift=2]at(0-1){2};
				\node[right,yshift=3]at(0-3){2};
				
				\node[above left,yshift=3]at(0-1-1){1};
				\node[above right,yshift=3]at(0-1-2){1};
				\node[above right,yshift=3]at(0-1-3){1};
				\node[above right,yshift=2]at(0-1-2-2){2};
				\node[above left,yshift=2]at(0-1-2-1){2};
				\node[above right,yshift=2]at(0-1-2-3){2};
				\end{tikzpicture}
				}
				\caption{An illustrative example where $a_0=$ CC in the PD. Row (player 1) moves first and Column (player 2) moves second.}
				\label{fig:illustrative1}
			\end{figure}

	\subsubsection{The no-harm principle}
	\label{sub:NHP_definition}
	
	Our specification of the no-harm principle applies when a player stays. We make this assumption for two reasons. First, in a dynamic strategic setting, the classical liberal has no reason to be concerned with the properties of any transitional (i.e., non-reference point) states in the extensive form game, particularly if they are purely mental constructs. Second, in contrast when a player chooses to stay, this matters for everyone because either another player follows this by choosing to stay and this becomes the implemented outcome; or, in so far as the play moves to another outcome, the reference point changes through the stya decision and this conditions future play and  the eventual  outcome. Indeed, an individual can only influence the character of the eventual outcome either directly or indirectly by choosing to stay because this changes the reference point. The point is that the only other way that a terminal node is reached is by mutual decisions to pass, in which case the original or prior reference point is implemented and the decision to pass has not affected the implemented outcome.

	\begin{definition}[No-harm principle]
	\label{def:NHP_action}
	Let $\Gamma$ be a game, $\sigma$ a strategy profile, $x\in X$ a non-terminal node, $b\in R(\sigma)$ the reference point at $x$, and $S(x)=a$. Action $\sigma_i(x)=a_i$ is said to satisfy the no-harm principle (NHP) at $x$ if for every $j\neq i$, $u_j (a)\geq u_j(b)$. 
	Strategy profile $\sigma$ satisfies the NHP at $x_0$ if for every $i$ and every $x$ as defined above, $\sigma_i(x)=a_i$ satisfies the NHP. Finally, strategy profile $\sigma$ satisfies the \textit{no-harm principle} if for every non-terminal $x'$, $(\sigma|x')$ satisfies the NHP at $x'$.
	\end{definition}
	
	In plain words, a player's stay action satisfies the NHP if their decision does not harm others with respect to the current reference point (e.g., see Figure~\ref{fig:PD_illustration2}). Accordingly, a strategy profile satisfies the NHP if every player's every stay action (both on-path and off-path) under that strategy profile satisfies the NHP. 
	
	Of note, the no-harm principle implies neither Pareto optimality nor even Pareto improvement from a reference point. In section~\ref{subsec:roleof_NHP}, we illustrate that assuming the no-harm principle may lead a society to a Pareto inferior outcome compared to the initial reference point. Even in situations in which the NHP leads to a Pareto improvement, the outcome of the game may be Pareto dominated as we illustrate in section~\ref{sub:structure_normal_form}.
	
	\subsection{The no-harm equilibrium}
	\label{sec:solution_concept}
	
	We assume that players are individually rational and farsighted in the usual sense of subgame perfection and are additionally constrained by the no-harm principle (NHP) in their action choices in $\Gamma$. Moreover, we assume that $G$, $\Gamma$, and the previous sentence are common knowledge \citep{lewis1969,aumann1976}. First, we define subgame perfect equilibrium \citep{selten1965,nash1951}.
	
	A pure strategy profile $\sigma \in \Sigma $ in game $\Gamma$ is called a \textit{subgame perfect equilibrium} (SPE) if for every player $i$ and for every non-terminal $x\in X$ where $i=I(x)$, $u_i(\sigma|x)\geq u_i(\sigma'_i,\sigma_{-i}|x)$ for every $\sigma'_{i}|x\in \Sigma_i|x$. Put differently, $\sigma$ is a subgame perfect equilibrium if it constitutes a Nash equilibrium in every subgame of $\Gamma$.
	
	\begin{definition}[No-harm equilibrium]
		\label{def:NHE}
		Let $G=(A,u)$ be a game. A pure strategy profile $\sigma^* \in \Sigma $ that satisfies the no-harm principle is called a \textit{no-harm equilibrium} (NHE) in $G$ if for every player $i$ and for every non-terminal $x\in X$ where $i=I(x)$ \[
		u_i(\sigma^*|x)\geq u_i(\sigma'_i,\sigma^*_{-i}|x)
		\]
		for every $\sigma'_{i}|x\in \Sigma_i|x$ such that $(\sigma'_{i},\sigma_{-i}^*)\in \Sigma$ satisfies the no-harm principle.
	\end{definition}
	
	In plain words, a strategy profile is an NHE if at every node the active player plays a best response under the constraint of the no-harm principle. Like subgame perfect equilibrium, in finite games no-harm equilibria can be computed using backward induction under the constraint of the NHP. Note that an NHE is \textit{not} equivalent to a strategy profile that is both a subgame perfect equilibrium and satisfies the no-harm principle, in part because in general there may be no SPE that satisfies the NHP, but as we show in section~\ref{sec:existence} an NHE always exists.
	
	Am NHE in $G$ depends, of course, on $\Gamma(a_0,k,I)$, i.e., the initial reference point, $a_0$, the player function $I$, and $k$. But for now it is important to note that the NHP \textit{per se} does not require Pareto optimality of the outcome. Players simply act independently and maximize their individual utility; they do not act to maximize the pay-offs of others. They can stay wherever they want as long as the outcome does not harm others with respect to the reference point and there could always be other outcomes that are as good for the individual who decides to `stay' and which would be better for the other players. We illustrate this point in section~\ref{subsec:roleof_NHP} with an example where the NHP by itself does not produce a Pareto efficient outcome (see also section~\ref{sub:structure_normal_form}). We next show under what conditions the NHE outcomes are Pareto optimal in $n$-person games.
	
	\section{Existence, uniqueness, and efficiency}
	\label{sec:existence}
	
	In this section, we first show that the NHE exists under general conditions in normal form games.
	
	\begin{theorem}[Existence]
		\label{thm:existence}
		Let $G=(A,u)$ be a game. For every initial reference point $a_0\in A$, for every $k\in \mathbb{N}^+$, and every player function $I$, there exists an NHE associated to $a_0$ in pure strategies.
	\end{theorem}
	
	\begin{proof}
		We fix an initial reference point $a_0$, a player function $I$, and some $k\in \mathbb{N}^+$.
		
		Notice that for every $a_0$, the game $\Gamma$ always possesses a pure subgame perfect equilibrium. This is true because $\Gamma$ is a well-defined finite extensive form game with perfect information. To see this, notice that the root of the game is $x_0$ where $S(x_0)=a_0$ and that every player function $I$ gives a unique player at every non-terminal node by construction of $\Gamma$. Because there are finitely many players and that $k$ is finite, the game $\Gamma$ ends after finitely many steps. This implies that there is always a subgame perfect equilibrium in pure strategies.
		
		Next, we assume that players act according to the NHP, which essentially puts a constraint on their choices in $\Gamma$. This implies that they have fewer (finitely many) choices under the NHP than they have under $\Gamma$. Because the NHP is common knowledge, the constrained game---i.e., the game in which all strategy profiles satisfy the NHP---is still of perfect information. Let $\sigma^*$ be a subgame perfect equilibrium in the constrained game, which exists by the same arguments as above. We note that $\sigma^*$ is an NHE in $\Gamma$ because $\sigma^*$ satisfies the NHP and at every node every active player plays a best response among the profiles that satisfy the NHP, since by construction all those profiles satisfy the NHP. This concludes the proof that $\sigma^*$ is an NHE.
	\end{proof}
	
	We next show under what conditions the uniqueness of the NHE outcome is guaranteed from an initial reference point.
	
	\begin{theorem}[Uniqueness]
		\label{thm:uniqueness}
		For every initial reference point $a_0\in A$, for every $k\in \mathbb{N}^+$, and every player function $I$, the NHE outcome associated to $a_0$ is unique.
	\end{theorem}
	
	\begin{proof}
		Given an initial reference point $a_0$, a finite $k$, and a player function $I$, the associated $\Gamma$ possesses a pure subgame perfect equilibrium as shown in the proof of Theorem~\ref{thm:existence}. We next show that this subgame perfect equilibrium outcome is unique. The reason is that no matter which player moves on a non-terminal node either (i) the player has a unique pure best response or (ii) the pure best responses all lead to the same outcome because the preferences of the players are strict in $G$. Thus, the subgame perfect equilibrium outcome in $\Gamma$ must be unique. Analogously, the subgame perfect equilibrium outcome in $\Gamma$ which is constrained by the NHP must also have a unique outcome. Together with Theorem~\ref{thm:existence}, this implies that the NHE outcome must be unique.
	\end{proof}
	
	Finally, we illustrate the relationship between the no-harm principle, rationality, and efficiency in $n$-person normal form games.
	
	\begin{theorem}[Efficiency]
	\label{thm:pareto}
	Let $G=(A,u)$ be a game. For every Pareto optimal outcome $a\in A$ there exists an initial reference point $a_0\in A$ such that for every $k\in \mathbb{N}^+$ and every player function $I$, the associated NHE outcome is $a$. Conversely, for every initial reference point $a_0\in A$, every $k\in \mathbb{N}^+$, every player function $I$, the associated NHE outcome is Pareto optimal. 
	\end{theorem}
	
	The proof of this theorem is in the Appendix~\ref{appendix}. Here we give an informal sketch of the proof.
	
	Given a player function $I$ and $k\in \mathbb{N}^+$, we first show that if an initial reference point $a_0$ is Pareto optimal then it is the NHE outcome from $a_0$. By way of contradiction, suppose that $a'\neq a_0$ is the NHE outcome. It implies that there exists at least one player who chose to stay (i.e., changed the reference point) in the path of play of an NHE, $\sigma^*$. Every player who did \textit{not} stay receives a strictly greater pay-off at $a'$ than $a_0$ because $\sigma^*$ satisfies the NHP. In addition, every player who did stay must, due to farsighted rationality and the NHP, receive a strictly greater pay-off at $a'$ than $a_0$. As a result, $a'$ Pareto dominates $a_0$, which contradicts the supposition that $a_0$ is Pareto optimal. 

    Second, we show that for an initial reference point $a_0$ that is not Pareto optimal, the NHE from $a_0$ must be Pareto optimal. By way of contradiction, suppose that $b$ is the NHE outcome from $a_0$ and $b$ is Pareto dominated by some action profile $a\neq b$. Let $\sigma^*$ be a NHE from $a_0$ such that  the first time a player stays at $b$ on the path of play of $\sigma^*$, the next player (say, $i$) also stays at $b$ by choosing $b_i$, hence terminating the game. Notice that if (i) there is a path from $b$ to $a$ along which the NHP is satisfied, then $b_i$ cannot be a best response of player $i$ because (ii) every player (including $i$) receives a strictly greater pay-off at $a$ than $b$ by our supposition that $b$ is Pareto dominated by $a$, and
    (iii) no other player can stay at an action profile which harms player $i$ along the path to $a$ since the NHP applies and $b$ is the reference point. We next show that (i) is true. First, notice that players can reach from $b$ to $a$ in at most $n$ moves by the following path of play. At every node, the active player $i$ plays move $a_{i}$ except when $a_{i}=b_{i}$, in which case player $i$ plays $p$ (i.e., pass). The NHP is not violated along this path of play because no player stays. Second, this path of play does not overlap with the path of play of $\sigma^*$ because if it did, then the active player at the overlapping node would have a profitable deviation to the path towards $a$. The reason is that if the state of the overlapping node $x$ is $a$, then the active player would have a profitable deviation from $\sigma^*$ to stay at $a$ and make $a$ the reference point because they are strictly better off at $a$ and staying at $a$ satisfies the NHP. By backward induction, the two paths of play cannot include an immediate predecessor $x'$ of node $x$ because the active player at $x'$ would have a profitable deviation to $x$, where the next player would stay. By analogous backward induction reasoning, one can conclude that the path of play from $b$ to $a$ and the path of play of $\sigma^*$ have an empty intersection. Thus, statement (i) holds as well. As desired, we reach a contradiction: $b$ cannot be the NHE outcome from $a_0$. 
	
	\subsection{Discussion of the assumptions}
	\label{subsec:discussion_results}
	
	We next discuss how different assumptions in the definition of $\Gamma$ and the NHE affect the results.
	
	\subsubsection{The no-harm principle}
	\label{subsec:roleof_NHP}
	
	To see why the NHP is essential for Theorem~\ref{thm:pareto}, first notice that the NHE definition would reduce to subgame perfect equilibrium if the no-harm principle were not assumed. Consider the following simple example and suppose that the NHP is \textit{not} assumed.
 
\begin{table}[h!]
	\begin{center}
  \setlength{\aboverulesep}{0pt}
\setlength{\belowrulesep}{0pt}
		\begin{tabular}{|c|c|c|}
			\midrule
			& L & R \\
			\midrule
			L & $4,3$ & $1,4$ \\
			\midrule
			R &	$2,1$ & $3,2$ \\
			\midrule
		\end{tabular}
	\end{center}
\end{table}
	Let the initial reference point be (1,4). Suppose that Row moves first, Column moves second, and this order strictly alternates. Row would not choose to pass at (1,4) because Column would then choose to pass too, making (1,4) as the outcome. On grounds of farsighted rationality, Row's best response is to move from (1,4) to (3,2), where Column as well as Row would stay, making it the  outcome. To see this, first notice that Column would not gain by moving to (2,1) from (3,2) because Row would not move to (4,3) as Row anticipates that Column would then go back to (1,4) where Row would have to either stay or pass because $k=1$. If Row stays at (1,4), then Column would simply make (1,4) the outcome by staying too. If Row passes at (1,4), then Column would also pass, making (1,4) the outcome. Second, notice that Row would not move back to (1,4) from (3,2), because Column would then stay there. Thus, without the no-harm principle and starting at (1,4), players would end up at (3,2), and this is Pareto dominated by (4,3).

	\subsubsection{Farsighted rationality}
	\label{subsec:rationality}
	
	Farsighted rationality is also a necessary assumption for Theorem~\ref{thm:pareto} because a strategy profile might satisfy the no-harm principle alone and yield a Pareto inferior outcome with respect to the initial reference point. The following $2\times 2$ game provides a simple example.
 
\begin{table}[h!]
	\begin{center}
  \setlength{\aboverulesep}{0pt}
\setlength{\belowrulesep}{0pt}
		\begin{tabular}{|c|c|c|}
			\midrule
			& L & R \\
			\midrule
			L & $2,2$ & $0,3$ \\
			\midrule
			R &	$1,0$ & $4,4$ \\
			\midrule
		\end{tabular}
	\end{center}\end{table}
	Suppose that the reference point is (2,2), Column moves first, Row moves second, and this order strictly alternates. Consider the strategy profile in which Column moves from (2,2) to (0,3) where Row stays, making (0,3) the updated reference point. Row's choice of L satisfies the no-harm principle since it does not harm Column player. Next, Column moves back to (2,2) and Row moves to (1,0) where first Column stays and then Row stays, making (1,0) the outcome. Column's decision to stay at (1,0) satisfies the no-harm principle since it does not harm Row player with respect to the updated reference point (0,3). Anticipating this and if the players were farsightedly rational, Column would \textit{not} stay at (1,0). But in the absence of the assumption of rationality, the aforementioned moves cannot be ruled out and it results in an outcome, (1,0), that is strictly Pareto dominated by (2,2).
	
	\subsubsection{Normal form structure}
	\label{sub:structure_normal_form}

    The normal form structure of game $G$ is also necessary for Theorem~\ref{thm:pareto}. We now assume both the NHP and farsighted rationality and illustrate this with the extensive form game given in Figure~\ref{fig:3player_forced}.
	
	Suppose that the initial reference point pay-off profile is (1,1) and player 1 moves first. There is a unique NHE in this game and it is Pareto dominated. To see this, notice that the best response of player 1 is to choose m, moving to (3,2) because if player 1 chooses to pass (i.e., p), then the best response of player 2 would be to move to (2,4), which would satisfy the NHP with respect to (1,1). Thus, player 1 moves to (3,2), where player 2 stays (or passes). The NHE outcome (3,2) coincides with the SPE outcome in this game. Although (3,2) is a Pareto improvement over the reference point (1,1), it is Pareto dominated by (4,3).

	One reason why the NHP and farsighted rationality of players do not immediately imply Pareto optimality is that the NHP puts a mild constraint on the behaviour of players. It restricts players from causing harm to others relative to the reference point, but beyond that the NHP does not require players to maximize the pay-off of others.

\begin{figure}
	\centering
	\resizebox{.93\textwidth}{!}{
	\begin{tikzpicture}[scale=.85, font=\footnotesize,edge from parent/.style={draw,thick}, decoration={markings, mark=at position 0.5 with {\arrow{>}}}
	]
	\tikzstyle{solid node}=[circle,draw,inner sep=1.2,fill=black];
	\tikzstyle{hollow node}=[circle,draw,inner sep=1.2];
	\tikzstyle{level 1}=[level distance=15mm,sibling distance=50mm]
	\tikzstyle{level 2}=[level distance=15mm,sibling distance=25mm]
	\tikzstyle{level 3}=[level distance=15mm,sibling distance=15mm]
	\node(0)[hollow node]{}
	child{node[solid node]{}
		child{node[below]{$\begin{pmatrix} 3\\2\end{pmatrix}$}
			edge from parent[postaction={decorate}, black, solid, very thick] node[above left]{s}
		}
		child{node[below]{$\begin{pmatrix} 3\\2\end{pmatrix}$}
			edge from parent[black, solid, thick] node[above right, yshift=-2]{p}
		}
		edge from parent[postaction={decorate}, black, solid, very thick] node[above right, yshift=2]{m}
	}
	child{node[below]{$\begin{pmatrix} 1\\1\end{pmatrix}$}
		edge from parent[thick] node[above right]{s}
	}
	child{node[solid node]{}
		child{node[below]{$\begin{pmatrix} 4\\3\end{pmatrix}$}
			edge from parent[thick] node[above right, yshift=3]{a}
		}
		child{node[below]{$\begin{pmatrix} 2\\4\end{pmatrix}$}
			edge from parent[postaction={decorate}, black, solid, very thick] node[left]{b}
		}
		child{node[below]{$\begin{pmatrix} 1\\1\end{pmatrix}$}
			edge from parent[thick] node[right]{s}
		}
		child{node[below]{$\begin{pmatrix} 1\\1\end{pmatrix}$}
			edge from parent[thick] node[above right, yshift=-2]{p}
		}
		edge from parent[thick] node[above right,yshift=-2]{p}
	}
	;
	\node[above,yshift=2]at(0){1};
	
	\node[above,yshift=2]at(0-1){2};
	\node[above,yshift=2]at(0-3){2};
	
	\end{tikzpicture}
	}
			\caption{An extensive form game in which the NHE outcome from reference point (1,1) is (3,2), which is Pareto dominated.}
\label{fig:3player_forced}
\end{figure}
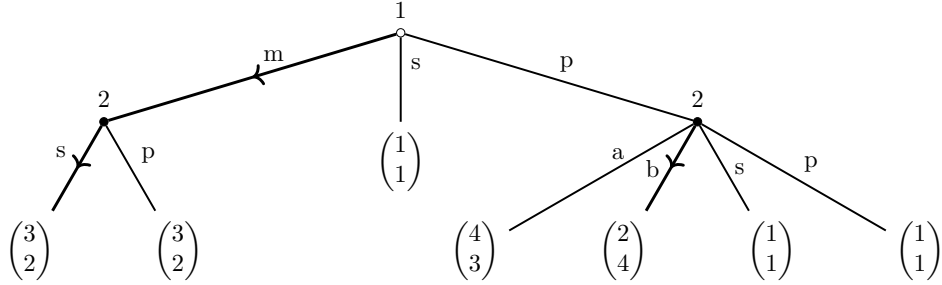

	\subsubsection{Unilateral termination}
	\label{sub:termination}
	
	In this section, we consider the modification of our model where each player has the opportunity to unilaterally terminate the game. In that case, our results would remain valid as long as the NHP applies to `termination' decisions as well. Consider the model presented in section~\ref{sec:NHP} with the modifications outlined below, holding everything else fixed. For every $i$ and every non-terminal $x$, let $A_i(x)\cup \{t\}$ be player $i$'s set of available actions at node $x$. If a player plays action $t$ at $x$, then the terminal node is reached and the outcome is defined as $a$ where $a=S(x)$. In section~\ref{sub:reference_points}, drop the terminal node condition (i) where the game terminates if two players choose to stay at a state. Add the following line to Definition~\ref{def:NHP_action}. Action $\sigma_i(x)=t$ is said to satisfy the no-harm principle (NHP) at $x$ if for every $j\neq i$, $u_j (a)\geq u_j(b)$. Under this modification of our model, the NHE associated with an initial node $a_0$ may differ from the NHE under the original setup. However, all three of our theorems would remain valid for the analogous reasons to the ones used in the proofs of respective theorems.
	
	One could also consider the following modification to our model in section~\ref{sec:NHP}. Suppose that the game terminates if $m$ players ($n\geq m>2$) choose to stay at a state instead of two players as is assumed in condition (i) in section~\ref{sub:reference_points}. This modification would not affect the application of the main arguments in the proofs of the three theorems. Thus, the theorems would remain valid in under this modification too.
	
	\subsubsection{Strict vs weak preferences}
	\label{subsec:roleof_preferences}
			
	One might wonder what happens to the no-harm equilibria when there are indifferences between the outcomes in $G$. In that case, Theorem~\ref{thm:existence} would remain valid, though Theorem~\ref{thm:uniqueness} would no longer hold. This is because subgame perfect equilibrium outcomes in $\Gamma$ need not be unique, which implies that NHE outcomes need not be unique either. For analogous reasons as in the proof of Theorem~\ref{thm:pareto} we can conclude that every Pareto optimal profile must be an NHE outcome, and an NHE outcome cannot be strictly Pareto dominated. Moreover, for every initial reference point, for every $k$, and for every $I$ there would always be an NHE that is Pareto optimal.
	
	\subsubsection{The player function and $k$}
	\label{subsec:roleof_player}
			
	While the three theorems hold for any $k\geq 1$ and any player function $I$, the associated NHE would potentially be different for different $k$ and $I$ (see, e.g., the example in subsection~\ref{subsec:3person}). However, this does not change the conclusion of, e.g., Theorem~\ref{thm:pareto} that any such NHE is Pareto efficient.
		
	In section~\ref{sec:NHP}, we put a restriction on player function $I$ that in every path of play each player has more or less equal number of nodes at which they are active. We next show that Theorem~\ref{thm:pareto} would not hold in general if we let the player function be arbitrary. Let $I'$ be a player function such that for every non-terminal node $x\in X$, $I'(x)=1$. Clearly, Theorem~\ref{thm:pareto} would not hold if the player function were $I'$. To see this, consider the PD with the initial reference point DD. Then, player 1 cannot by themself move to CC. Thus, player 1 would stay at DD, which is Pareto dominated.
		
	A different way to interpret the player function $I$ is that it may be chosen by Nature in the beginning of the game according to the stochastic process described below. Fix a game $G=(A,u)$ and action profile $a_0\in A$. Let $q\in \Delta N$ be a probability distribution over the set of players $N$ such that for every player $i$, $q(i)>0$ and $\sum_{i}q(i)=1$. For a given natural number $k\in \mathbb{N}^+=\{1,2,...\}$, and probability distribution $q\in \Delta N$, we define an associated extensive form game with perfect information and Nature move denoted by $\Gamma'(a_0, k, q)=(N, X' , I', u', S', H')$. At the root, $x'_0$, of $\Gamma'(a_0, k, q)$, Nature randomly chooses a player function $I$, and then players play the game $\Gamma(a_0, k, I)=(N, X , I, u, S, H)$.
	
	Let $I':X'\to N$ denote the player function in $\Gamma'(a_0, k, q)$, where $I'(x)$ is the active player at node $x\in X'$. At $x'_0$, Nature chooses player function $I:X\to N$, where $X\subset X'$, according to the following process. The probability player $j$ is the active player at a non-terminal node $x_m$ is given by the conditional probability $P(j|x_{m-1})$, where $x_{m-1}$ is the immediate predecessor of $x_m$, which is defined as follows. Let $\underline{m}=floor(\frac{m}{n})$. If $x_{\underline{m}n+1}=m$, then $P(j|x_{m-1})=q(j)$. If $x_{\underline{m}n+1}<m$, then 	
	\[
	P(j|x_{m-1})=
		\begin{cases}
			0, ~\text{if}~j= I'(x_{\underline{m}n+1}),~\text{or}~j=I'(x_{\underline{m}n+2}),...,~\text{or}~j=I'(x_{m-1})\\
			\frac{q(j)}{\sum_{i} q(i)-\sum_{i=\underline{m}n+1}^{m-1} q(I(x_{i}))},~\text{else}.
		\end{cases}
	\]
	
	Notice that the player function $I$ defined as above satisfies the restriction we put in section~\ref{sec:NHP}. Thus, irrespective of the realisation of Nature's randomisation, the three theorems would remain valid in $\Gamma(a_0, k, I)$.

	\section{Illustrations}
	\label{sec:illustrations}
	
	\subsection{The Prisoners' Dilemma}
	
	We first go back to the PD. It follows CC, CD and DC are Pareto efficient for some reference points and so all are NHEs, but DD is not Pareto efficient and is not an NHE.  Nevertheless, although it is clear DD is not Pareto efficient, it is perhaps not immediately obvious why deviation from DD satisfies both (1) NHP and (2) farsighted rationality. 
	
	Let $\Gamma(a_0, k, I)$ be the extensive form game that begins with the reference point $a_0=$ DD. Suppose that Row moves first, Column moves second, and this order strictly alternates. Consider the deviation from DD by Row to C. This deviation may seem to be precluded because it does not immediately satisfy Row's farsighted rationality---since Row is worse off---at CD. Nevertheless, to see whether it might satisfy Row's farsighted rationality (2), we need to consider what Column does at CD because CD may not be the stopping point. Indeed, Column cannot stay at CD because CD harms Row relative to the reference point of DD and so will not satisfy the no-harm principle (1). Hence if Column were to find themselves at CD, they would have to move and CC is the only option. Will Row stay at CC? CC satisfies (1) the NHP. It also satisfies Row's farsighted rationality (2) because a move to DC would produce a `cycle' back to DD from which no further deviation would be permitted because $k=1$. Thus since CC is better for Row than DD, it is also farsightedly rational for Row to stay at CC, which becomes the outcome of the NHE. In other words, you have to trace through what happens with a deviation by Row using DD as the reference point before you can see that (2) is also satisfied by the deviation of Row to C from reference point DD; and DD is not an NHE. Instead, CC is the NHE associated with the reference point of DD.
	
	As mentioned earlier, CD and DC are also NHE outcomes in the PD; but they are only NHEs when respectively the initial reference points are CD and DC.
	
	\subsection{Stag-Hunt and Hawk and Dove}
	
	It is well-known that Pareto optimality and the Nash equilibrium are logically distinct concepts in the sense that neither concept is a refinement of the other. As we show in Theorem~\ref{thm:pareto} the NHEs coincide with Pareto optimal profiles. Thus, there is no logical relationship between the set of NHEs and the set of Nash equilibria. Two further illustrations in this sub-section bring this out. In the Stag-Hunt game, NHE is a case of Nash refinement, and in the Hawk-Dove game, NHE expands the Nash equilibria; whereas, as we have just seen, in the PD the Nash equilibrium is not an NHE. 
	
	Consider, first, the Stag-Hunt game:
\begin{table}[h!]
	\begin{center}
  \setlength{\aboverulesep}{0pt}
\setlength{\belowrulesep}{0pt}
		\begin{tabular}{|c|c|c|}
			\midrule
			& Stag & Hare \\
			\midrule
			Stag & $4,4$ & $1,3$ \\
			\midrule
			Hare &	$3,1$ & $2,2$ \\
			\midrule
		\end{tabular}
	\end{center}
	\end{table}
 
	Clearly, irrespective of the reference point the players will end up at (Stag, Stag), which is the Pareto dominant profile and also a Nash equilibrium. (Hare, Hare) is a Nash equilibrium but not an NHE.
 
	Next, in the Hawk and Dove (Chicken) game, (Dove, Dove) is a NHE as well as the two Nash equilibria (H,D) and (D,H):
		
  \begin{table}[h!]
	\begin{center}
  \setlength{\aboverulesep}{0pt}
\setlength{\belowrulesep}{0pt}
		\begin{tabular}{|c|c|c|}
			\midrule
			& Hawk & Dove \\
			\midrule
			Hawk & $1,1$ & $4,2$ \\
			\midrule
			Dove &	$2,4$ & $3,3$ \\
			\midrule
		\end{tabular}
	\end{center}
		\end{table}
		
	This is an interesting game in that both non-myopic equilibrium and NHE predictions coincide. The two concepts in general give different predictions mainly due to the no-harm principle. In general, not every non-myopic equilibrium is an NHE such as (D,D) in the PD. Conversely, not every NHE is a non-myopic equilibrium because not every Pareto optimal outcome is a non-myopic equilibrium such as (A,D) in game 22 (Brams, 1994), which is given below. 
		
  \begin{table}[h!]
	\begin{center}
  \setlength{\aboverulesep}{0pt}
\setlength{\belowrulesep}{0pt}
		\begin{tabular}{|c|c|c|}
			\midrule
			& C & D \\
			\midrule
			A & $2,4$ & $3,3$ \\
			\midrule
			B &	$1,2$ & $4,1$ \\
			\midrule
		\end{tabular}
	\end{center}
		\end{table}
		
	The Pareto optimal profiles in this game are (A,D), (B,D), and (A,C), which is the non-myopic equilibrium.

	\subsection{A three-person illustrative example}
	\label{subsec:3person}
	
	We next illustrate the no-harm equilibria in a three-person game presented in Figure~\ref{fig:3player}. Throughout this example, we assume that the initial reference point is (A,D,E).

	Assume that Row moves first, Column second, Matrix third, and this order strictly alternates. Figure~\ref{fig:alternatingoffers} illustrates part of the game tree where the arrows show the on-path moves of the NHE, which can be described as follows. Row moves to (8, 8, 4), and Column stays at (8, 8, 4), which makes it the reference point. A best response of Matrix is to stay at (8, 8, 4), making it the outcome of the NHE from (3, 1, 2). Notice that Matrix can move to (4, 4, 5), but cannot stay in matrix F because this would violate the NHP with respect to the reference point (8, 8, 4).
	\begin{figure}
		\begin{center}
   \setlength{\aboverulesep}{0pt}
\setlength{\belowrulesep}{0pt}
			\begin{tabular}{|c|c|c|}
				\midrule
				E & C & D \\
				\midrule
				A & $1, 6, 1$ & $3, 1, 2$ \\
				\midrule
				B &	$2, 7, 3$ & $8, 8, 4$ \\
				\midrule
			\end{tabular}
			\quad
			\begin{tabular}{|c|c|c|}
				\midrule
				F & C & D \\
				\midrule
				A & $5, 2, 6$ & $7, 5, 8$ \\
				\midrule
				B &	$6, 3, 7$ & $4, 4, 5$ \\
				\midrule
			\end{tabular}
		\end{center}
		\caption{No-harm equilibria in a three-person illustrative game}
		\label{fig:3player}
	\end{figure}
 
    Now, assume that Matrix moves first, Column second, Row third, and this order strictly alternates. The initial  reference point is (3, 1, 2) as before. We explain the on-path actions of the NHE as follows. Matrix moves to (7, 5, 8) by playing F, where both Column and then Row stay. The reason why it is a best response for Column to stay at (7, 5, 8) is that (i) Column receives their highest pay-off in matrix F, (ii) Matrix player would prefer to stay at any outcome in matrix F rather than moving to matrix E, and (iii) every outcome in matrix F satisfies the NHP with respect to the reference point (3, 1, 2). For analogous reasons, it is also a best response for Row to stay at (7, 5, 8). Thus, (7, 5, 8) is the outcome of the NHE from the initial reference point (3, 1, 2).

		\begin{figure}
			\centering
			\resizebox{.96\textwidth}{!}{
			\begin{tikzpicture}[font=\footnotesize,edge from parent/.style={draw,thick}, decoration={
						markings,
						mark=at position 0.5 with {\arrow{>}}}
			]
			\tikzstyle{solid node}=[circle,draw,inner sep=1.2,fill=black];
			\tikzstyle{hollow node}=[circle,draw,inner sep=1.2];
			\tikzstyle{level 1}=[level distance=15mm,sibling distance=50mm]
			\tikzstyle{level 2}=[level distance=15mm,sibling distance=25mm]
			\tikzstyle{level 3}=[level distance=15mm,sibling distance=15mm]
			\node(0)[hollow node]{}
			child{node[solid node]{}
				child{node[solid node]{}
					child{node[below]{} edge from parent[shorten >=0.7cm, dotted, thick] node[left]{}}
					edge from parent[black, solid, thick] node[above left]{C}
				}
				child{node[solid node][below]{}
						child{node[solid node]{}
							child{node[below]{} edge from parent[shorten >=0.7cm, dotted, thick] node[left]{}}
							edge from parent[black, solid, thick] node[above left]{F}
						}
						child{node[below]{$\begin{pmatrix} 8\\8\\4\end{pmatrix}$}
							edge from parent[postaction={decorate}, black, solid, very thick] node[below left]{E}
						}
						child{node[solid node][below]{}
							child{node[below]{} edge from parent[shorten >=0.7cm, dotted, thick] node[left]{}}
							edge from parent[black, solid, thick] node[above right]{$p$}
						}					
					edge from parent[postaction={decorate}, black, solid, very thick] node[below left]{D}
				}
				child{node[solid node][below]{}
					child{node[below]{} edge from parent[shorten >=0.7cm, dotted, thick] node[left]{}}
					edge from parent[black, solid, thick] node[above right]{$p$}
				}
				edge from parent[postaction={decorate}, black, solid, very thick] node[above left, yshift=-3]{B}
			}
			child{node[solid node, below]{}
				child{{}
					edge from parent[shorten >=0.9cm, dotted, thick] node[above left]{}
				}
				edge from parent[thick] node[right]{A}
			}
			child{node[solid node]{}
				child{{}
					edge from parent[shorten >=0.9cm, dotted, thick] node[above left]{}
				}
				edge from parent[thick] node[above right, yshift=-4]{$p$}
			}
		;
			\node[above,yshift=2]at(0){$R$};
			
			\node[above,yshift=2]at(0-1){$C$};
			\node[right,yshift=3]at(0-2){$C$};
			\node[right,yshift=3]at(0-3){$C$};
			
			\node[left,yshift=3]at(0-1-1){$M$};
			\node[right,yshift=3]at(0-1-2){$M$};
			\node[right,yshift=3]at(0-1-3){$M$};
			\end{tikzpicture}
			}
			\caption{Part of the game tree of $\Gamma$ presented in Figure~\ref{fig:3player} where the arrows illustrate the NHE path. Row moves first, Column second, and Matrix third. (The full game tree is not shown due to space constraints.)}
			\label{fig:alternatingoffers}
		\end{figure}	
	
	At the outset, it looks like Row and Column should be able to implement their most preferred outcome (8, 8, 4) in the game. However, as shown above this is not possible if Matrix is the first-mover at the initial reference point. This three-person example illustrates that the player function $I$ can affect the NHE associated with any reference point, but $I$ does not affect the conclusion that the NHEs are Pareto optimal.

	\section{Discussion}
	\label{sec:discussion}
	
	\subsection{Rule-like constraints on individual action}

The no-harm principle is a rule-like constraint on individual action that someone who believes in or subscribes to classic liberalism will wish to follow. This contrasts with the models where individuals are altruistic or have other kinds of social preference that can transform a PD and predict CC. In these social preference models, an individual typically personally values the material pay-offs enjoyed by others and this enjoyment typically grows/falls with the size of the material pay-off to others. By taking account of these other regarding preferences, a PD game in terms of material pay-offs is transformed into a different game in terms of the utility number pay-offs; and CC can become a Nash equilibrium in this transformed utility pay-off version of the interaction. With the no-harm principle, we begin with utility pay-offs so as to have an encompassing definition of harm and there is no analogous relationship whereby one person’s utility varies with that of another’s. Instead, an action either satisfies the no-harm principle or it does not and in the one case the action is permissible and in the other it is not. In this respect the no-harm principle is akin to a version of rule rationality: that is, if the language of preference satisfaction is retained individuals have a lexicographic preference for following a rule(s) and so when they act to satisfy their preferences, they act in accordance with the rule(s).
	
	Another example of rule rationality is provided by Kant's categorical imperative: to `act only according to that maxim whereby you can at the same will that it should become a universal law'. The rule-like constraint in this instance is that the action is universalized (and is evaluated under this constraint independently of whether others actually take the same action). It is well-known that Kantian rule rationality can produce similar results to those we have derived for the no-harm principle in prisoners' dilemma interactions (e.g. see Roemer, \citeyear{roemer2010}, in the economics literature and for a more general discussion, O'Neill, \citeyear{oneill1989}). The difference is that the Kantian rule has a more controversial connection to classical liberalism than the no-harm principle (e.g. see Berlin, \citeyear{berlin1969}); and, to our knowledge, the no harm principle has not been studied before in non-cooperative game theory, whereas the Kantian one has (e.g. see Roemer, \citeyear{roemer2010}). 
	
	\subsection{No-harm principle in social choice theory}
	
	The no-harm principle is related in classical liberal political philosophy to the presumption that the State should not intervene in individual decision making when the consequences of those decisions apply only to the individual(s) making the decisions. This  non-intervention principle has, of course, since \cite{sen1970} featured prominently in the social choice literature. The no-harm principle has also been used more recently in this social choice literature (see, e.g., Lombardi et al., \citeyear{lombardi2016}). Mariotti and Veneziani (\citeyear{mariotti2009}; \citeyear{mariotti2013}; \citeyear{mariotti2020}) introduce a notion called ``Non-Interference'' principle which roughly says that society's preferences should not change following a change in circumstances that affect only one individual and for which everyone else is indifferent.  Recently, \cite{mariotti2020} show that there is inconsistency between their ``Non-Interference'' principle and the Pareto principle (i.e., if everyone in a society prefers an alternative $x$ to $y$, then society should prefer $x$ to $y$) in a non-dictatorship.
	
	Our formalization of the NHP differs from Mariotti and Veneziani's in two main respects, the framework and the conceptual definition. The most obvious difference between the two principles is that ours applies to actions within a game theoretical framework whereas theirs applies to the preferences within a social choice context. Conceptually, under our no-harm principle a player is allowed to choose any action as long as this action does not eventually harm (and may benefit) other players with respect to the reference point. However, the ``Non-Interference'' principle does not apply to a change in social situations that leave some members of the society better off.

	Although we share the interest in the implications of subscribing to the tenets of classical liberalism with the social choice literature, the approach here is very different.  We are not interested in the implications of classical liberalism for a social planner---as is the case in the social choice literature. Instead, we are concerned with how the introduction of the no-harm principle as a constraint on individual decision making in games affects the equilibrium outcomes of those games.

	\subsection{Pareto efficiency}
	We can interpret our framework as a set of necessary and sufficient assumptions that give a non-cooperative characterization to Pareto efficiency via the NHP.
	
	In a recent and related development, \cite{che2020} provided a characterization of the Pareto optima via utilitarian welfare maximization. While both our and their approaches are sequential in nature, the main difference between the two papers is that their framework is non-strategic whereas we provide a non-cooperative foundation for Pareto optimality via the no-harm principle.
	
	\cite{ray2020} recently introduced a general class of games called ``games of love and hate'' to describe strategic situations with pay-off-based externalities. They show that every Nash equilibrium in this class is Pareto optimal under some regularity conditions; though, not every Pareto optimal profile is a Nash equilibrium. Note that both their and our approaches ensure that all equilibria are Pareto efficient. \cite{ray2020} make certain assumptions on the pay-off functions of the players (e.g., prisoners' dilemma falls outside of that class), whereas we assume that players maximize utility subject to the rule-like constraint of the no-harm principle.
	
	In modeling the no-harm principle, we have followed the approach of \cite{brams1994}. In this same tradition, \cite{brams2021} show that there is always a non-myopic equilibrium that is Pareto optimal; though, not all Pareto optimal profiles are non-myopic equilibria, and not all non-myopic equilibria are Pareto optimal, as we discussed in section~\ref{sec:illustrations}. The main difference in our results comes from the no-harm principle that we assume, which restricts the actions of the players to ones that satisfy the well-known principle of classical liberalism.
	
	\subsection{Welfare economics}
	\label{subsec:welfare}
	
	There is an analogy between Theorem~\ref{thm:pareto} and the first and the second welfare theorems. The first fundamental theorem of welfare economics states that, roughly speaking, irrespective of the set of initial endowments the competitive  equilibrium is Pareto efficient. In our setting, Theorem~\ref{thm:pareto} says that for every initial reference point in the society the associated NHE is always Pareto efficient. The second fundamental theorem of welfare economics states that any Pareto efficient allocation can be achieved as a competitive equilibrium allocation for some set of initial endowments. Analogously, by Theorem~\ref{thm:pareto} for every Pareto efficient profile in a strategic game there is an initial reference point for which this is the  NHE. 
	
	While these theorems point to a similar outcome, the processes and assumptions behind them are quite different.  There is no account of how a competitive equilibrium is reached in general equilibrium theory and, critically, a competitive general equilibrium assumes that agents are price takers. In short, there is no strategic interaction in the sense of game theory. 
	
	\subsection{Related frameworks in the literature}
		
	In this sub-section, we discuss the seemingly unrelated fields to which our framework is closely related. To the best of our knowledge, the NHP is not studied in these frameworks.
		
	\subsubsection*{Theory of Moves}
	As mentioned earlier, our framework is closely related to Brams's (1994) seminal work. Starting from \cite{brams1981}, non-myopic equilibrium has been developed and extended under different assumptions and domains; see, e.g., \cite{kilgour1984}.
	
	\subsubsection*{Oligopolistic markets}
	\cite{marschak1978} study oligopolistic markets in which there is an initial status quo of price/quantity within a normal-form game, and firms can unilaterally change their actions starting from this status quo, observe the history of changes, and react to them. Like in our framework, firms only care about the final outcome, ignoring any transitory profits.
	
	\subsubsection*{Stochastic games}
	\cite{shapley1953} first introduced this well-studied class of non-cooperative games. Our framework is most related to a sub-class of stochastic games called recursive games of perfect information in which there are a finite number of states where each player is ``active'' at some state and can ``switch'' to another state (see, e.g., Flesch, Kuipers, Schoenmakers, and Vrieze,  \citeyear{flesch2010}, and the references therein). At each state, the active player has the option to either switch from the current state to another state or ``quit,'' in which case the state payoffs are collected. For a similar class of games called Dynkin games, see, e.g., \cite{solan2003}.
	
	\subsubsection*{Farsightedness in cooperative games}
	Since the ground-breaking book of \cite{neumann1944}, cooperative games have been applied in many contexts. There is a sub-field of cooperative games in which, like in our framework, a strategy profile in a normal-form game can be considered as a state. A coalition, which may be an individual, can `move' from one state to another for which they are `effective'. One can then define myopic/farsighted notions of core and stable sets, abstracting away from strategic considerations. It is impossible to do this literature justice, but for a non-exhaustive list, see, for example, \cite{moulin1982,greenberg1990,ray2015,koray2018}, and \cite{bloch2021}. We refer the interested reader to \cite{bloch2021} who provide a comprehensive review of the relevant literature in this sub-field.
	
	\subsubsection*{Alternating-offers bargaining and cheap talk in normal-form games}
	
	Our framework relates to non-cooperative cheap talk and bargaining games in which the set of alternatives corresponds to the set action profiles in a normal-form game. Players can move sequentially and make proposals, which can then be accepted or rejected. The normal-form structure underlying these games distinguishes them from Rubinstein's (\citeyear{rubinstein1982}) seminal bargaining game with a discount rate. For a non-exhaustive list, see, e.g., \cite{kalai1981}, \cite{farrell1988}, and \cite{santos2000}. For a thorough literature review, we refer the interested reader to \cite{fukuda2022} and the references therein.
	
	\section{Conclusion}
	\label{sec:conclusion}
	
	Game theory standardly makes no assumption about what motivates individuals to act other than they have preferences they seek to satisfy. While this is an admirably parsimonious assumption, it is also misleading when people either subscribe to the political philosophy of classical liberalism or live in a society that is legally founded on the principles of classical liberalism. Such people are additionally constrained, either legally or by their own beliefs, by the no-harm principle. This is because the principle is the key constraint placed on the exercise of individual freedom by J. S. Mill in his classic manifesto for individual liberty: \textit{On Liberty}. Thus, for those who live in a classically liberal society and/or who believe in classical liberalism, a question naturally arises: how is behaviour in games affected by the additional individual constraint on action supplied by the no-harm principle? We offer part of an answer to this question. 
	
	We show with our operationalization of the no-harm principle that this addition dramatically alters the predictions regarding what happens in games. In particular, the no-harm equilibria are always Pareto optimal. This stands in marked contrast to standard game theory where there is no necessary connection between Nash equilibria and Pareto optimality. It is important in the derivation of this result to note that our operationalization of the no-harm principle does not require Pareto optimality; nor does it even secure Pareto improvements from a starting position. It is the combination of farsighted rationality with the no-harm principle that secures Pareto optimality. 
	
	Our paper opens up two main directions for future theoretical research. First, what are the other frameworks in which strategic foundations of Pareto optimality can be studied? Second, while the definitions of no-harm principle and the NHE can be extended to games under incomplete and imperfect information in a straightforward way as subgame perfect equilibrium is well-defined under these settings, can it be extended to games with infinite horizons? 
	
	It also suggests an important new direction for empirical research that has possible implications for public policy. It is well known from experiments, for example, that some people behave selfishly in public goods/PD interactions and others behave pro-socially by contributing to the public good. The pro-social contributions are typically understood through the prism of social preferences and selfishness is understood through the absence of such preferences. To what extent, then, might they be better understood through the differing sway or influence that the no-harm principle has on individuals? In particular, while the puzzle from these experiments from the perspective of standard game theory has centred on why subjects contributed anything to the public good, it changes with the result of this paper. Rather, the puzzling question becomes: why do so many subjects in these experiments, when they come from liberal societies, behave selfishly?\footnote{\cite{amadae2016} has an answer to this question: the rise and influence of Game Theory. In fact, the seeds of the analysis in this paper were sown by \cite{amadae2016} and \cite{heap2016}. Amadae (2016) argues that game theory has encouraged a form of neo-liberalism that is distinct from classical liberalism precisely because game theory dispenses with the no-harm principle. She conjectures that the no-harm principle would dramatically alter the prediction of what rational individuals would do in a prisoners' dilemma. Hargreaves Heap (2016) reviewed this book and found this conjecture intuitively plausible and so, in effect, reproduced it in the review.} This, in turn, connects  to the new policy agenda suggested by this paper: we need, when considering policy interventions, to understand better and focus on the circumstances under which people are not guided by the no-harm principle in liberal democratic societies.
	
	To put this last point slightly differently, it is often argued that the prisoners' dilemma helps explain why people decide to restrict their freedoms (e.g. Hobbes, \citeyear{hobbes1651}). From the perspective of this paper, it is not the occurrence of prisoners' dilemmas in social and economic life, a sort of brute fact about some types of interactions, that occasions this retreat from liberty. Rather, it is a retreat from the liberal conception of liberty that is responsible for making prisoners' dilemma interactions problematic; and from a policy perspective, it is important to get the source of the problem right.
	\section*{Appendix}\label{appendix}
	
	\subsection*{Proof of Theorem~\ref{thm:pareto}}
	\label{sec:proof_thm3}
	
	\noindent \textit{Proof of the first part.}
	We first show that if an initial reference point $a_0$ is Pareto optimal, then for every $k$, and every player function $I$, the associated NHE outcome is $a_0$ in game $\Gamma(a_0, k, I)$, proving the first part of the theorem. 
	
	To reach a contradiction, suppose that $a_0$ is not the NHE outcome, and the outcome of an NHE $\sigma^*$ from $a_0$ is given by some $a'\neq a_0$. We know that the NHE outcome from each initial reference point is unique by Theorem~\ref{thm:uniqueness}. Because $a'\neq a_0$ and $a'$ is the outcome of $\sigma^*$, it must be that $R(\sigma^*)\setminus \{a_0\}$ is non-empty. Too see this, suppose that $R(\sigma^*)=\{a_0\}$. It implies that $a'= a_0$, which is a contradiction. Thus, there exists at least one player who stayed along the path of play of $\sigma^*$, that is, changed a reference point in $R(\sigma^*)$. For every player $j$ who did \textit{not} stay along the path of play of $\sigma^*$, it must be that $u_j(a')>u_j(a_0)$. This is because for every reference point $a\in R(\sigma^*)$,  $u_j(a)\geq u_j(a_0)$ due to the fact that $\sigma^*$ satisfies the NHP. For every player $i$ who did stay along the path of play of $\sigma^*$, it must be that $u_i(a')>u_i(a_0)$ due to two reasons. 
	
	First, player $i$'s pay-off cannot be diminished before $i$ stays because every preceding stay decision must satisfy the NHP. To see this, let $x\in X$ be the node such that player $i$ stays for the first time at $a\in R(\sigma^*)$, and $b$ be the reference point at $x$. Then, it must be that $u_i(b)> u_i(a_0)$ due to the NHP, that is, other players could not have stayed and harmed player $i$.
	
	Second, it would not be optimal for player $i$ to stay by playing $a_i$ unless $u_i(a')>u_i(a_0)$. That is, if $i$ stays at $a$ where $S(x)=a$, then $i$ eventually must benefit from this action due to farsighted rationality and the NHP. If $a_i$ is optimal at $x$, i.e., $\sigma^*_i(x)=a_i$, then $u_i(a')\geq u_i(b)$ because otherwise player $i$ would \textit{not} stay at $a$, changing the reference point $b$. Notice that $i$ can move to another state or play $p$, in which case the minimum pay-off $i$ would receive is $u_i(b)$. This is because (i) if someone else stays at a state different than $b$, then $i$ cannot be harmed, and (ii) if everyone passes, then the outcome would be $b$. But we also have that $u_i(b)> u_i(a_0)$. Therefore, $u_i(a')> u_i(a_0)$.
	
	As a result, it implies that for every player $i'$,  the inequality $u_{i'}(a')>u_{i'}(a_0)$ is satisfied, irrespective of whether $i'$ stays or not along the path of play of $\sigma^*$. This contradicts to our supposition that $a_0$ is Pareto optimal. Therefore, $a_0$ must be the outcome of $\sigma^*$.\qed 
	
	\vspace{0.3cm}
	\noindent \textit{Proof of the second part.}
	Next, we show that for every initial reference point $a_0\in A$, every $k\in \mathbb{N}^+$, and every player function $I$, the associated NHE outcome in game $\Gamma(a_0, k, I)$ is Pareto optimal. In the first part of the proof we already showed that if $a_0$ is Pareto optimal, the associated NHE outcome is Pareto optimal. It is left to show that for an initial reference point $a_0$ that is not Pareto optimal, the NHE associated with $a_0$ must be Pareto optimal. Let $\sigma$ be an NHE from $a_0$, which exists by Theorem~\ref{thm:existence}.
	
	To reach a contradiction, suppose that the outcome of $\sigma$ is $b$, and $b$ is Pareto dominated by some action profile $a\neq b$. We obtain a contradiction in two main steps.
	
	\textbf{Step 1:} Given $\sigma$, we construct an NHE $\sigma^*$ from $a_0$ such that there exists a player $i$ who chooses $b_{i}$ at some node $y$ in $[\sigma^*]$ and makes the reference point, $b\in R(\sigma^*)$, at $y$ the  outcome. Since $b$ is the outcome of $\sigma$,  the outcome of any NHE from $a_0$ must be $b$ by Theorem~\ref{thm:uniqueness}. 
	
	Consider path of the play, $[\sigma]$, of $\sigma$ excluding the terminal node $\overline{[\sigma]}$. Note that whether the active player at the penultimate node in $[\sigma]$ stayed or passed,  $b$ must have been the reference point at some point during the path of the play of $\sigma$. Let $i'$ be the player who makes $b$ the reference point for the first time at some node $y'\in[\sigma^*]$. Let $y\in[\sigma^*]$ be an immediate successor of $y'$ such that $\sigma_{i'}(y')\rightarrow y$ and $i=I(y)$ be the player who moves at $y$. Since the outcome of $\sigma$ is $b$, it must be optimal (i.e., a best response under the constraint of the NHP) for player $i$ to stay at $y$ and make $S(y)=b$ the  outcome. Player $i$'s stay action, $b_i$, is available at node $y$ because $b$ is the reference point for the first time at $y'$. In addition, choosing $b_i$ clearly satisfies the NHP because $b$ is already the reference point at $y$. We then construct $\sigma^*$ such that $\sigma^*_{i}(y)=b_{i}$. If $\sigma_{i}(y)= b_{i}$, then define $\sigma^*=\sigma$. If $\sigma_{i}(y)\neq b_{i}$, then for every player $m$ and every non-terminal node $\hat{x}\in X\setminus \{y\}$ define $\sigma^*_{m}(\hat{x})=\sigma_{m}(\hat{x})$. As desired, we have constructed an NHE $\sigma^*$ from $a_0$ such that there is player $i$ who makes the reference point, $b\in R(\sigma^*)$, at $y$ the  outcome in the first opportunity.

	\textbf{Step 2:} We next show that $\sigma^*$ and hence $\sigma$ cannot actually be an NHE because $\sigma^*_{i}(y)$ cannot be player $i$'s optimal choice at $y$. In other words, player $i$ has a unilateral profitable deviation from $\sigma^*$ and this deviations satisfies the NHP. Notice that if	(i) there exists a path of play from $b$ to $a$ along which the NHP is satisfied, then $\sigma^*_{i}(y)$ cannot be optimal because
	(ii) for every player $m$ (including $i$) $u_m(a)>u_m(b)$, and
	(iii) no other player can stay at an action profile which harms player $i$ along the path because $b$ is the reference point.
	
	We first show (i). Let $b'\in A$ and $a'\in A$ be two action profiles. We first show that for any $k\geq 1$ and any player function $I$, there is always a path of play between $b'$ and $a'$. Let $[b',a']$ be the path of play from $b'$ to $a'$ with the following property. For every node $y'$ in this path of play, the active player at $y'$, $i'=I(y')$, chooses $a_{i'}$ except when $a_{i'}=b_{i'}$, in which case player $i'$ chooses to pass, $p$. Notice that if the players follow this path, then the play would reach to $a'$ from $b'$ in at most $n$ moves. Note that no player stays along the path of play, so no action in the constructed path of play violates the NHP.
	
	Now, let $[b,a]$ be the path of play from $b$ to $a$ constructed as above. We next show that $[b,a]\cap [\sigma^*]=\emptyset$, i.e., the constructed path of play does not overlap with the path of play of $\sigma^*$. In other words, we make sure that $\sigma^*$ does not prescribe players to choose actions at some nodes in $[\sigma^*]$ such that these actions then make the constructed path, $[b,a]$, infeasible due to e.g. $k=1$. 
	
	Let $x \in [b,a]$ be a node such that $S(x)=a$. Then, it must be that $x \notin [\sigma^*]$ because if $x \in [\sigma^*]$, then it would be a unilateral profitable deviation for the active player at $x$ to stay at $a$ and make $a$ the reference point. This is because $u_m(a)>u_m(b)$ and staying at $a$ satisfies the NHP. To see why staying at $a$ satisfies the NHP, suppose (to reach a contradiction) that there exists a player $\hat{i}$ such that $u_{\hat{i}}(a)<u_{\hat{i}}(\hat{a})$, where $\hat{a}$ is the reference point at $x$. But we know that the outcome of $\sigma^*$ is $b$ and that $u_{\hat{i}}(b)<u_{\hat{i}}(a)$, which implies that $u_{\hat{i}}(b)<u_{\hat{i}}(\hat{a})$. Thus, either player $\hat{i}$ harms themself by staying at $b$, or someone else harms $\hat{i}$. It implies that either the farsighted rationality of $\hat{i}$ is violated or the NHP is violated, a contradiction. As a result, if there exists a node $x \in [\sigma^*]$ such that $S(x)=a$, then the active player would stay at $a$, making it the reference point. But if $a$ is the reference point, then $b$ cannot be the outcome of the NHE $\sigma^*$ because every player is strictly better off at $a$. This leads to a contradiction to our supposition that the outcome of $\sigma^*$ is $b$. This establishes that $x \notin [\sigma^*]$.
	
	Let $y\in [b,a]$ be an immediate predecessor of $x$. Then, it must be that $y \notin [\sigma^*]$ because if $y \in [\sigma^*]$, then the player at $y$ would have a unilateral profitable deviation by moving to $a$, anticipating that the next player would stay at  $a$ as shown in the previous paragraph. Note that this deviation would not violate the NHP by construction of $[b,a]$. Next, let $y'\in [b,a]$ be a (not necessarily immediate) predecessor of $y$. By backward induction, notice that $y' \notin [\sigma^*]$ because if $y' \in [\sigma^*]$, then the player $I(y')$ would have a unilateral profitable deviation from $\sigma^*_{I(y')}(y')$ by playing an action that leads to a node $y''\in [b,a]$, where $y''$ is an immediate successor of $y'$. Thus, the paths of play $[b,a]$ and $[\sigma^*]$ have an empty intersection. As a result, statement (i) holds: there is a path of play from $b$ to $a$ such that no player violates the NHP, given the path of play of $\sigma^*$. 
	
	Statement (ii) holds by our supposition that $b$ is Pareto dominated. 
	Statement (iii) holds by definition of the NHP:
	player $i$ has a unilateral profitable deviation at $b$ by playing $a_i$ unless $b_i=a_i$, in which case $i$ has a unilateral profitable deviation by playing $p$, entering the path from $b$ to $a$ as constructed above. By the NHP, no other player in the path can reduce $i$'s pay-off. In addition, no player stays in the path from $b$ to $a$. Thus, player $i$ will eventually receive a strictly greater pay-off by deviating to the constructed path because $u_i(a)>u_i(b)$.
	
	As a result, if the outcome $b$ of NHE $\sigma^*$ is Pareto dominated by some $a$, then player $i$ who stays at $b$ and make $b$ the outcome would have a unilateral profitable deviation from $\sigma^*$. This contradicts to our supposition that the NHE outcome from $a_0$ is $b$. As desired, this implies that the NHE outcome from any $a_0$ must be Pareto optimal.\qed
	\bibliographystyle{chicago}

\bibliography{references}

\end{document}